\documentclass[a4paper,12pt]{article}
\usepackage{amsmath}
\usepackage{latexsym}
\usepackage{arydshln}
\usepackage{ascmac}
\usepackage{amsfonts}
\usepackage{amssymb}
\usepackage{latexsym}
\def\qed{\hfill $\Box$}
\usepackage{amsthm}

\theoremstyle{definition}
\newtheorem{thm}{Theorem}[section]
\newtheorem{defi}{Definition}[section]
\newtheorem{prop}{Proposition}[section]
\newtheorem{rem}{Remark}[section]

\newtheorem{cor}{Corollary}[section]

\newtheorem*{pf}{{\it Proof}}
\title{
\bf A certain generalization of $q$-hypergeometric functions \\
and their related  monodromy preserving deformation II
}
\date{}
\author{Kanam Park\\ {\it Departmant of Mathematics, Graduate School of Science, Kobe University},\\{\it 1-1, Rokkodai, Nada-ku, Kobe} 657-8501, {\it Japan}\\Email: kpaku@math.kobe-u.ac.jp}

  \makeatletter
    
    \@addtoreset{equation}{section}
  \makeatother
\begin{document}
\maketitle

\begin{abstract}
We define a nonlinear $q$-difference system $\mathcal{P}_{N,(M_-,M_+)}$ as monodromy preserving deformations of a certain linear equation.  
We study its relation to a series $\mathcal{F}_{N,M}$ defined as a certain generalization of $q$-hypergeometric functions.
\end{abstract}
\section{Introduction}
%

In the previous work \cite{park18},  we defined a series $\mathcal{F}_{N,M}$   
\begin{equation}
\mathcal{F}_{N,M} \Big({\{a_j\},\{ b_i\} \atop \{ c_j \}}; \{ y_i\}\Big)=\displaystyle\sum_{m_i\geq 0}\prod_{j=1}^N \cfrac{(a _j)_{|m|} }{(c_j)_{|m|}}\prod _{i=1}^M \cfrac{(b_i)_{m_i}}{(q)_{m_i}}\prod_{i=1}^M y_i ^{m_i},
\end{equation}
as a certain generalization of $q$-hypergeometric functions, 
where,  $(a)_n=\frac{(a)_{\infty}}{(q^n a)_{\infty}}$, $(a)_{\infty}=\prod_{i=0}^{\infty} (1-q^i a)$ and $0<|q|<1$. 
The aim of this paper is to study the system of $q$-difference non-linear equations which admits a particular solution in terms of the function $\mathcal{F}_{N,M}$.
This problem was solved in cases of $(N,M)=(1,M)$ in \cite{park18}, and we will achieve our goal for all $(N,M)$ in this paper. 
We remark that these results can be considered as a natural $q$-analog of Tsuda's results \cite{tsuda12}.
 
%
   The contents of this paper is as follows. 
   In the next section, we formulate a system of nonlinear $q$-difference equations $\mathcal{P}_{N,(M_-,M_+)}$ which is the monodromy preserving deformation we aim at.
   In section $3$, we review some facts on the series $\mathcal{F}_{N,M}$ from the previous work \cite{park18}. 
   In section 4, we derive a Pfaffian system which the series $\mathcal{F}_{N,M}$ from an integral representation of it. 
   In section 5, we construct the Pfaffian system which $\mathcal{F}_{N,M}$ satisfies and derive solution of the system $\mathcal{P}_{N+1,(M,M)}$. \section{A monodromy preserving deformation $\mathcal{P}_{N,(M_-,M_+)}$}\label{mdnm}
In this section, we define a monodromy preserving deformation $\mathcal{P}_{N,(M_-,M_+)}$.   
Its relation to a generalization of $q$-hypergeometric functions $\mathcal{F}_{N,M}$  will be given in the next section.
\begin{defi}
 We define nonlinear $q$-difference system $\mathcal{P}_{N,(\varepsilon_1,\varepsilon_2,\cdots, \varepsilon_M)}$, $\varepsilon_i=\pm 1$ as monodromy preserving deformations of the following  linear equation
\begin{equation}\label{lax}
\begin{array}{l}
\Psi (qz)=\Psi(z)A(z), \quad A(z)=DX_1^{\varepsilon_1}(z)X_2^{\varepsilon_2}(z)\cdots X_M^{\varepsilon_M}(z),\\
\end{array}
\end{equation} 
where  the matrices $D$ and $X_i^{\varepsilon_i}(z)$ stand for 
\begin{equation}
\begin{array}{c}
D=\mathrm{diag}[d_1,d_2,\cdots,d_N],\\
X_i(z)=\mathrm{diag} [u_{1,i},u_{2,i},\cdots u_{N,i}]+\Lambda, \vspace{6pt}\\
\Lambda=
\begin{bmatrix}
  &1&         &\\
  & &\ddots&\\
  & &          &1\\
z&  &&
\end{bmatrix}, \quad
\end{array}
\end{equation}
and $\{ u_{j,i}\}$ are dependent variables and $\{ d_j, c_i\} $ are  parameters such that
\begin{equation}
\prod_{j=1}^N u_{j,i}=c_i.
\end{equation} 
\end{defi}
As is shown by the following Proposition, except for the dependence on deformation direction, the equation \eqref{lax} 
system $\mathcal{P}_{N,(\varepsilon_1,\varepsilon_2,\cdots, \varepsilon_M)}$ 
essentialy depends only on  $M_-$, $M_+$, where $M_{\pm}=\# \{ \varepsilon_i | \varepsilon_i=\pm 1 \}$. Namely, $A(z)$ for different $\vec{\varepsilon}$ are equivalent 
if $M_-$, $M_+$ are the same. Therefore we sometimes use notation $\mathcal{P}_{N,(\varepsilon_1,\varepsilon_2,\cdots, \varepsilon_M)}=\mathcal{P}_{N,(M_-,M_+)}$. 
\begin{prop}
For any permutation $\sigma \in S_M$, there exists a unique set of variables $\{ u_{j,i}' \}$ such that the following equation holds 
\begin{equation}\label{xpmt}
\begin{array}{c}
X_1^{\varepsilon_1}(z)X_2^{\varepsilon_2}(z)\cdots X_M^{\varepsilon_M}(z)={X_{\sigma(1)}^{\varepsilon_{\sigma(1)}}}'{X_{\sigma(2)}^{\varepsilon_{\sigma(2)}}}'\cdots {X_{\sigma(M)}^{\varepsilon_{\sigma(M)}}}',\\
 {\rm det} X_k={\rm det} X_k'=z-c_k,
   \end{array}
\end{equation}
where the matrix $X_i'$ stands for a matrix $X_i$ whose variables $u_{j,i}$ are replaced by $u_{j,i}'$. 
\end{prop}
\begin{pf} 

We show the existence first. 
Define a transformation $s_{k,l}^{\varepsilon_k,\varepsilon_{l}}$:
\begin{equation}\label{gkn0}
s_{k,l}^{\varepsilon_k,\varepsilon_l}(u_{j,k},u_{j,l})=(u_{j,k}\frac{Q_{j+1,k}^{\varepsilon_k,\varepsilon_{l}}}{Q_{j,k}^{\varepsilon_k,\varepsilon_l}},u_{j,l}\frac{Q_{j,k}^{\varepsilon_k,\varepsilon_l}}{Q_{j+1,k}^{\varepsilon_k,\varepsilon_l}}),
\end{equation}
where $Q_{j,i}$ is polynomial in $u_{j,i}$, $u_{j,i+1}$ given by
\begin{equation}\label{gkn}
\begin{array}{l}
Q_{j,i}^{+,+}=\displaystyle\sum_{a=1}^N (\prod_{k=1}^{a-1} u_{j+k,i} \prod_{k=a+1}^N u_{j+k,i+1}),\\
 Q_{j,i}^{-,-}=\left.\displaystyle\frac{1}{Q_{j,i}}\right|_{i \leftrightarrow i+1},\vspace{6pt}\\
Q_{j,i}^{+,-}=u_{j,i}-u_{j,i+1},\\
Q_{j,i}^{-,+}=(u_{j-1,i}-u_{j-1,i+1})^{-1}.
\end{array}
\end{equation}
Such defined a birational mapping $s_{k,l}^{\varepsilon_k,\varepsilon_l}$ satisfies 
$s_{k,l}^{\varepsilon_k,\varepsilon_l}(X_k^{\varepsilon_k}X_l^{\varepsilon_l})
={X_l^{\varepsilon_l}}' {X_k^{\varepsilon_k}}'$. 
For any permutation $\sigma \in S_M$,  
by composing  the equations in \eqref{gkn0}, we obtain a birational mapping 
 which satisfies \eqref{xpmt}. 

To show the uniqueness, we note that the kernel of the left hand side of the first equation \eqref{xpmt} at $z=c_{\sigma(M)}$ should be equal to that of the right hand side. This condition determine the matrix $X_{\sigma(M)}^{\varepsilon_{\sigma(M)}}$ uniquely. 
\qed
\end{pf}
\begin{rem}\label{gar}
It is easily shown that the system $\mathcal{P}_{2,(0,2M)}$ and $\mathcal{P}_{2N+2,(0,2)}$  are  respectively equivalent to the $2(M-1)$ dimensional $q$-Garnier system  \cite{sakai05} and the $q$-$P_{(N+1,N+1)}$ \cite{suzuki15} in the sense that the matrix $A$ in \eqref{lax} is equivalent to the matrix $A$ given in \cite{ny17}. 
\end{rem}
\begin{rem}
The equation \eqref{lax} for the system $\mathcal{P}_{N,(M_-,M_+)}$ is an extended  version of the system in  \cite{kny02} \cite{or16} obtained as the $N$-reduced $q$-KP hierarchy through a similarity condition. 
\end{rem}
  \section{The series  $\mathcal{F}_{N,M}$ as an extension of $q$-hypergeometric functions}\label{part1}
  In this subsection, we will recall some results in \cite{park18} which are needed in the following subsections. 
    \label{fmn}\begin{defi} (\cite{park18}, Definition 2.1) We define a series $\mathcal{F}_{N,M}$ as
  \begin{gather} \label{fmn}
\mathcal{F}_{N,M} \Big({\{a_j\},\{ b_i\} \atop \{ c_j \}}; \{ y_i\}\Big)=\displaystyle\sum_{m_i\geq 0}\prod_{j=1}^N \cfrac{(a _j)_{|m|} }{(c_j)_{|m|}}\prod _{i=1}^M \cfrac{(b_i)_{m_i}}{(q)_{m_i}}\prod_{i=1}^M y_i ^{m_i},
\end{gather}
  where
   $(a)_n=\frac{(a)_{\infty}}{(q^n a)_{\infty}}$ and $0<|q|<1$. Here and in what follows the symbol $(a)_{\infty}$ means $(a)_{\infty}=\prod_{i=0}^{\infty} (1-q^i a)$. 
The series \eqref{fmn} converges in the region $|y_i|<1$ and  is continued analytically  to $|y_i|\geq1$.
  \end{defi}
 When $N=1$ or $M=1$, the series \eqref{fmn} is equal to the $q$-Appell-Lauricella function $\varphi_D$ or the generalized $q$-hypergeometric function ${}_{N+1} \varphi _N$, respectively \cite{gr04}:
\begin{equation}
\begin{array}{ll}
\mathcal{F}_{1,M} \displaystyle \left( {a,\{ b_i\} \atop c} ;\{ y_i \} \right)&=\displaystyle\sum_{m_i \geq 0} \frac{(a)_{|m|}}{(c)_{|m|}} \prod_{i=1}^M \frac{(b_i)_{m_i}}{(q)_{m_i}} y_i ^{m_i}
=\displaystyle\varphi_D \left( {a , \{b_i \} \atop c };\{ y_i \} \right),
\end{array}
\end{equation}
\begin{equation}
\begin{array}{ll}
\displaystyle \mathcal{F}_{N,1} \left( {\{a_j \} , b \atop \{c_j \} };y \right) &=\displaystyle\sum_{n \geq 0} \prod_{j=1}^N \frac{(a_j)_{m}(b)_{m}}{(c)_{m} (q)_{m}} y^{m} =\displaystyle{}_{N+1} \varphi_N \left( {\{a_j \} , b \atop \{ c_j \} }; y \right).
\end{array}
\end{equation}
There is a duality relation between the series $\mathcal{F}_{N,M}$ and $\mathcal{F}_{M,N}$ as follows:
\begin{prop}(\cite{park18}, Proposition 2.1)\label{dual}
The series $\mathcal{F}_{N,M}$ satisfies the relation
\begin{equation}\label{taishou}
\mathcal{F}_{N,M} \left( {\{y_j \} ,\{a_i \} \atop \{b_jy_j \}};\{x_i \} \right) =\prod_{j=1}^N \frac{(y_j)_{\infty}}{(b_jy_j)_{\infty} } \prod_{i=1}^M \frac{(a_i x_i )_{\infty}}{(x_i)_{\infty}} \mathcal{F}_{M,N} \left( {\{x_i \} , \{b_j \} \atop \{a_i x_i \} };\{y_j \} \right).
\end{equation}
\end{prop}
\begin{rem}
When $N=1$ or $M=1$, the relation \eqref{taishou} is known (see  \cite{And72} \cite{kn03} for example).
\end{rem}
We can interpret the equation \eqref{taishou} as
an integral representation of $\mathcal{F}_{N,M}$ as follows.
   \begin{cor}\label{cor:sekibun}(\cite{park18}, Corollary 2.1.)
  With $y_j=q^{\gamma_j}$, the relation \eqref{taishou} can be rewritten as
 \begin{equation}\label{sekibun}
  \mathcal{F}_{N,M} \left( {\{ q^{\gamma_j} \} ,\{a_i \} \atop \{b_j q^{\gamma_j} \}};\{x_i \} \right)=
  \prod_{j=1}^N \frac{(q^{\gamma_j},b_j)_{\infty}}{(b_j q^{\gamma_j} ,q)_{\infty}} \prod_{j=1}^N \int _0^1d_q t_j \prod_{i=1}^M \frac{(a_ix_i \prod_{j=1}^N t_j )_{\infty}}{(x_i \prod_{j=1}^N t_j)_{\infty}} \prod_{j=1}^N \frac{(qt_j)_{\infty}}{(b_jt_j)_{\infty}} \frac{t_j ^{\gamma_j-1}}{1-q},
  \end{equation}
  where 
  the Jackson integral is defined as
  \begin{equation}
  \int_{0}^c d_q t f(t) = c(1-q) \sum_{n\geq 0} f(cq^n) q^n.
  \end{equation}
  \end{cor}
     \begin{prop}\label{qde}(\cite{park18}, Proposition 2.2)
  The series $\mathcal{F}=\mathcal{F}_{N,M} \Big({\{a_j\},\{ b_i\} \atop \{ c_j \}}; \{ y_i\}\Big)$ satisfies the  $q$-difference equations
  \begin{equation}\label{sabun1}  
  \begin{array}{l}
  \left\{ \prod_{j=1}^N (1-c_jq^{-1} T) \cdot (1-T_{y_s}) -y_s  \prod_{j=1}^N (1-a_jT)\cdot (1-b_s T_{y_s})\right\}\mathcal{F}=0\quad(1\leq s \leq M),\\
  \{y_r (1-b_r T_{y_r})(1-T_{y_s}) -y_s (1-b_s T_{y_s})(1-T_{y_r}) \} \mathcal{F}=0\quad (1\leq r <s\leq M),
  \end{array}
  \end{equation}
  where $T_{y_s}$ is the $q$-shift operator for  the variable $y_s$ and $T=T_{y_1}\cdots T_{y_M}$.
  \end{prop}
  We consider a representation by a Pfaffian system of the equation \eqref{sabun1} in the next subsection.
  \section{A Pfaffian system derived from $\mathcal{F}_{N,M}$}\label{pfnm}
  In this section, we derive a Pfaffian system of size $(MN+1)\times (MN+1)$ from an integral representation of {$\mathcal{F}_{N,M}$}. From Corollary \ref{cor:sekibun}, the integral representation of {$\mathcal{F}_{N,M}$} is given in \eqref{sekibun}. 
 We compute the Pfaffian system for the integral \eqref{sekibun}.  
 To do this, we follow the method given in \cite{matsuo93}, \cite{mimachi89} in the same way as \cite{park18}.
  We denote the integrand of \eqref{sekibun} as $\Phi (\{ u_j \}_{j=1}^N)$
   \begin{equation}\label{integrand}
  \Phi(\{ u_j\}_{j=1}^N )=\prod_{j=1}^N{u_j ^{\nu_j}} \frac{(qu_j/u_{j-1})_{\infty}}{(b_ju_j/u_{j-1})_{\infty}}\prod_{i=1}^M \frac{(a_ix_i u_N )_{\infty}}{(x_i u_N)_{\infty}},
  \end{equation}
  where $u_j$ stands for $\prod_{k=1}^j t_k$ and
  we put the parameter $\nu_j:=\gamma_j-\gamma_{j+1}$ $(\nu_{N+1}=0)$. 
   We define functions $\Psi_0$, {$\Psi_{j,i}$ $(1\leq j \leq N, 1\leq i \leq M)$} as 
    \begin{equation}
    \begin{array}{c}
  \displaystyle\Psi_0=\langle \Phi p_0\rangle,\\
  \displaystyle\Psi_{j,i}=\langle \Phi p_{j,i} \rangle
   \end{array}
   \end{equation}
where $p_0=1$, $p_{j,i}=\displaystyle\frac{u_{j-1}-b_j u_j}{1-a_ix_iu_N}\prod_{k=1}^{i-1}\frac{1-x_ku_N}{1-a_kx_ku_N}$ 
    and $\langle \  \rangle$ means a kind of  Jackson integral for {$u_1,u_2,\cdots,u_N$}. 
  Namely, {$\langle f(\{u_j\}) \rangle=\displaystyle\sum_{n_j \in \mathbb{Z}} f(\{ q^{n_j}\})$.
  We will see  that $\{ \Psi_0, \Psi_{1,1},\cdots, \Psi_{N,M} \}$} is a basis of the solutions of the Pfaffian system. 
  \begin{rem}
  The basis $\{p_0, p_{j,i}\}$ defined above has been used in  many literatures 
  (See for example \cite{matsuo93},  \cite{mimachi96}, \cite{res92}). 
  Such a kind of basis is convenient for a computation because it is directly related to a shift of function 
  $\Phi$ as $\Phi p_{j,i}=u_{j-1}T_{a_i}T_{b_j}T_{x_{i-1}}T_{x_{i-2}}\cdots T_{x_1}\Phi$. 
  Therefore the following equations hold
  \begin{equation}\label{psiint}
  \begin{array}{c}
 \displaystyle \Psi_0= \frac{(\{b_jq^{\gamma_j}\},\{q\})_{\infty}}{(\{q^{\gamma_j} \},\{b_j\})_{\infty}}\mathcal{F}, \\ 
  \displaystyle \Psi_{j,i}
  =\frac{(1-b_j)}{(1-b_jq^{\gamma_j})}\prod_{l=1}^{j-1}\frac{(1-q^{\gamma_l})}{(1-b_lq^{\gamma_l})}
  T_{a_i}T_{x_1}T_{x_2}\cdots T_{x_{i-1}} T_{b_j}
  \mathcal{F},
   \end{array}
   \end{equation}
   \begin{equation}
  \mathcal{F}=\mathcal{F}_{N,M}\left( {\{ q^{\gamma_j} \} ,\{a_i \} \atop \{b_j q^{\gamma_j} \}};\{x_i \} \right)\quad ({1\leq j \leq N \atop 1 \leq i \leq M})
  \end{equation}
  \end{rem}
  We define an exchange operator $\sigma_i$ $(1 \leq i \leq M)$ acting on a function $f$ of $\{x_i, a_i \}$ as
  \begin{equation}
  \sigma_i (f) =f|_{x_i\leftrightarrow x_{i+1}, a_i\leftrightarrow a_{i+1}}.
  \end{equation}
  We note that $\sigma_i (\Phi (\{ u_j \}))=\Phi (\{u_j\})$. When the operator $\sigma_i$ acts on functions $p_0$, $p_{j,i}$ {$(1\leq j \leq N, 1\leq i \leq M)$}, we have the following relations.
  \begin{prop}\label{prop:gokan} For $1 \leq i \leq M$, we have
  \begin{equation}\label{gokan}
  \sigma_i(\Psi_{j,k})=
  \begin{cases}
  \displaystyle\frac{(1-a_i)x_i}{x_i-a_{i+1}x_{i+1}}\Psi_{j,i}+\frac{a_ix_i-a_{i+1}x_{i+1}}{x_i-a_{i+1}x_{i+1}}\Psi_{j,i+1},&(k =i)\\
  \displaystyle \frac{x_i-x_{i+1}}{x_i-a_{i+1}x_{i+1}}\Psi_{j,i}+\frac{(1-a_{i+1})x_{i+1}}{x_i-a_{i+1} x_{i+1}} \Psi_{j,i+1},& (k=i+1)\\
  \Psi_{j,k}. & (k \neq i,i+1)
  \end{cases}
  \end{equation}
  \end{prop}
  \begin{proof}   The last equation of \eqref{gokan} is obvious.
The other equations of \eqref{gokan} is proved in the same way as Proposition 3.1. in the previous work \cite{park18}.  
Namely we consider the equations
  \begin{equation}
  \begin{array}{c}
  \sigma_i (p_{j,i})=s_1\ p_{j,i}+s_2\ p_{j,i+1},\\
  \sigma_i (p_{j,i+1})=s_3\ p_{j,i}+s_4\ p_{j,i+1}.
  \end{array}
  \end{equation}
  There exist unique coefficients $s_1$, $s_2$, $s_3$, $s_4$ independent of $u_j$  satisfying these
relations.
  \end{proof}
  For the action of the $q$-shift operator $T_{x_M}$ on $\Psi_0$, $\cdots$, $\Psi_{N,M}$, we obtain the following equations.
  \begin{prop}\label{prop:mup} We have
  \begin{equation}\label{shift}
  \begin{array}{l}
  T_{x_M} (\Psi_0)=\displaystyle\frac{x_M-\beta
  }{a_Mx_M-\beta} \rho (\Psi_0)+\frac{(a_M-1)x_M}{a_Mx_M-\beta}\sum_{j=1}^N (\prod_{k=1}^{j-1} b_k ) \rho (\Psi_{j,1}), \\
  T_{x_M} (\Psi_{j,i})=\rho(\Psi_{j,i+1})\quad(i\neq M), \\
  T_{x_M} (\Psi_{j,M})=\displaystyle\frac{q^{-\sum_{k=j}^N \gamma_k}(1-b_j)}{a_Mx_M-\beta} 
  (\prod_{l=j+1}^N b_l ) \biggl[\rho (\Psi_0)-\sum_{k=1}^{j-1} (\prod_{t=1}^k b_t ) \rho (\Psi_{k,1}) \\
  \qquad \qquad \quad 
  +\displaystyle\frac{a_Mx_M-\beta/b_j  }{(1-b_j)\prod_{k=j}^N b_k} \rho (\Psi_{j,1}) -a_Mx_M \sum_{k=j+1}^N (\prod_{t=k}^N b_t ^{-1} ) \rho (\Psi_{k,1}) \biggr],
\end{array}
\end{equation}
  where $\beta=\prod_{j=1}^Nb_j$, $\rho=\sigma_{M-1}\cdots\sigma_1$.
    \end{prop}  
    \begin{pf}
  First, we make a shift by  $T_{x_M}$ on $\Phi p_0$, 
  $\Phi p_{j,i}$ $(1\leq j \leq N, 1\leq i \leq M)$. We easily obtain the following equations
  \begin{equation}\label{p0}
  \begin{array}{l}
    T_{x_M} (\Phi p_{j,i})=\Phi \rho (p_{j,i+1}) \quad (i \neq M),\\
  T_{x_M} (\Phi p_0) =\Phi \displaystyle\frac{1-x_Mu_N}{1-a_Mx_Mu_N},\\
  T_{x_M} T_{u_j}^{-1}T_{u_{j+1}}^{-1}\cdots T_{u_N}^{-1} (\Phi p_{j,M})  
  =\displaystyle q^{-\gamma_j} \Phi \frac{u_{j-1}-u_j}{1-a_Mx_Mu_N}.\\
  \end{array}
  \end{equation}
  The right-hand side of the second and third equations in \eqref{p0}, can be rewritten as a linear combination of $\rho(\Phi p_0)=\Phi p_0$ and 
  $\rho (\Phi p_{j,1})=\Phi \displaystyle\frac{u_{j-1}-b_ju_j}{1-a_Mx_Mu_N}$
  respectively, that is,
  \begin{equation}\label{0}
  \begin{array}{l}
  T_{x_M} (\Phi p_0)=\displaystyle\frac{x_M-\beta}{a_Mx_M-\beta} \rho (\Phi p_0)+\frac{(a_M-1)x_M}{a_Mx_M-\beta
  }\sum_{j=1}^N (\prod_{k=1}^{j-1} b_k ) \rho (\Phi p_{j,1}), 
  \\
  T_{x_M} (\Phi p_{j,M})=\displaystyle\frac{q^{-\sum_{k=j}^N \gamma_k}(1-b_j)}{a_Mx_M-\beta} (\prod_{l=j+1}^N b_l ) \biggl[\rho (\Phi p_0)-\sum_{k=1}^{j-1} (\prod_{t=1}^k b_t ) \rho (\Phi p_{k,1}) \\
  \qquad \qquad \quad +\displaystyle\frac{a_Mx_M-\beta}{(1-b_j)\prod_{k=j}^N b_k} \rho (\Phi p_{j,1}) -a_Mx_M \sum_{k=j+1}^N (\prod_{t=k}^N b_t ^{-1} ) \rho (\Phi p_{k,1}) \biggr],
\end{array}
\end{equation}
  where $\beta=\prod_{j=1}^Nb_j$, $\rho=\sigma_{M-1}\cdots\sigma_1$.
  Integrating the first equation of \eqref{p0} and equations in  \eqref{0}  with respect to  
  $u_1,u_2,\cdots,u_N$,
  we obtain the equations in \eqref{shift}.
  \qed
  \end{pf}
   Combining Proposition \ref{prop:gokan} and Proposition \ref{prop:mup}, we obtain the following Theorem. 
  \begin{thm}\label{thm:pfaff}
  The vector 
  $\overrightarrow{\Psi}=[\Psi_A]_{A\in I}$, $I=\{0\} \cup \{ (j,i)|1\leq j \leq N, 1\leq i \leq M\}$  
  satisfies the  Pfaffian system of rank $MN+1$
  \begin{equation}\label{appu}
  T_{x_M} \overrightarrow{\Psi}=\rho \mu \overrightarrow{\Psi}=\sigma_{M-1}\sigma_{M-2}\cdots \sigma_1 \mu \overrightarrow{\Psi},
  \end{equation}
  where the operators $\mu$ and $\sigma_i$ ($1\leq i \leq M$) are
  \begin{equation}\label{ms}
  \begin{array}{c}
  \mu (x_1)=qx_1, \quad \mu(\Psi_0)=\displaystyle\frac{x_1-\beta}{a_1x_1-\beta} \Psi_0+\frac{(a_1-1)x_1}{a_1x_1-\beta}\sum_{j=1}^N (\prod_{k=1}^{j-1} b_k ) \Psi_{j,1},\\
  \mu (\Psi_{k,l})=\begin{cases} \Psi_{k,l+1} \quad (l \neq M),\\
     \displaystyle\frac{q^{-\sum_{k=j}^N\gamma_k}(1-b_j)}{a_1x_1-\beta} 
  (\prod_{l=j+1}^N b_l )
   \biggl[\Psi_0-\sum_{k=1}^{j-1} (\prod_{t=1}^k b_t ) \Psi_{k,1}+\\
   \quad \displaystyle\frac{a_1x_1-\beta/b_j  }{(1-b_j)\prod_{k=j}^N b_k} \Psi_{j,1} -a_1x_1 \sum_{k=j+1}^N (\prod_{t=k}^N b_t ^{-1} ) \Psi_{k,1} \biggr] &(l=M),\end{cases}
  \end{array}
  \end{equation}
\begin{equation}\label{sgm}
\begin{array}{c}
\sigma_i(a_i)=a_{i+1},  \sigma_i(a_{i+1})=a_i,   \sigma_i(x_i)=x_{i+1}, \sigma_i(x_{i+1})=x_i, \vspace{\baselineskip}\\
\sigma_i (\Psi_0)=\Psi_0, \sigma_i(\Psi_{k,l})=\begin{cases}
  \displaystyle\frac{(1-a_i)x_i}{x_i-a_{i+1}x_{i+1}}\Psi_{k,i}+\frac{a_ix_i-a_{i+1}x_{i+1}}{x_i-a_{i+1}x_{i+1}}\Psi_{k,i+1},&(l =i)\\
  \displaystyle \frac{x_i-x_{i+1}}{x_i-a_{i+1}x_{i+1}}\Psi_{k,i}+\frac{(1-a_{i+1})x_{i+1}}{x_i-a_{i+1} x_{i+1}} \Psi_{k,i+1},& (l=i+1)\\
  \Psi_{k,l},& (l \neq i, i+1)
  \end{cases}
\end{array}
\end{equation}
\end{thm}
  \begin{proof}
 The results follows by direct computation using \eqref{gokan} \eqref{shift}.
  \end{proof}
 \begin{rem}
 The equations for the shift of the other variables 
 \begin{equation}\label{appui}
  T_{x_i}  \overrightarrow{\Psi}=\overrightarrow{\Psi}A_i,
  \end{equation}
  for $i=1,\cdots, M-1$
 can be derived from \eqref{appu} and the action of $\{ \sigma_i\}$ in Proposition \ref{prop:gokan}. For example, the equation for the shift of the variable $x_{M-1}$ is obtained as follows
 \begin{equation}
 T_{x_{M-1}}\overrightarrow{\Psi}=\sigma_{M-1}T_{x_M}\sigma_{M-1}\overrightarrow{\Psi}=\overrightarrow{\Psi}R_{M-1}\cdot \sigma_{M-1}(A_M) \cdot \sigma_{M-1}(R_{M-1}(qz)).
 \end{equation}
By construction, the coefficient matrices in \eqref{appu} and \eqref{appui} 
satisfy a compatibility condition
\begin{equation}
A_i (T_{x_i} A_j)=A_j(T_{x_j}A_i).
\end{equation}
\end{rem} 
 \section{A reduction to $(N+1)\times (N+1)$
  form}\label{redn}
 In this section we reduce the equation \eqref{appui} into $(N+1)\times (N+1)$  form. 
To do this, we specialize the parameter $a_M$ to be $1$. Then the integrand \eqref{integrand}, and  hence 
$\Psi_{j,i}$ $(1 \leq j \leq N, 1\leq i \leq M-1)$
 and $\Psi_0$, become independent of $x_M$.
  Therefore we can consider 
  $\Psi_{j,i}$ $(1\leq j \leq N, 1\leq i \leq M-1)$
   as $\Psi_0$ times 
   $r_{j,i}$, where 
   $r_{j,i}$
 is a rational function in $x_0$, $\cdots$, $x_{M-1}$. The explicit form of $r_{j,i}$ is as follows (see \eqref{psiint})
 \begin{equation}\label{rji}
 r_{j,i}=\frac{1-b_j}{1-q^{\gamma_j}}\prod_{l=1}^j \frac{1-q^{\gamma_l}}{1-b_l q^{\gamma_l}}\prod_{k=1}^N \frac{(q^{\gamma_k},b_k)_{\infty}}{(b_k q^{\gamma_k},q)_{\infty}}\frac{T_{a_i}T_{x_1}T_{x_2}\cdots T_{x_{i-1}} T_{b_j}\mathcal{F}}{\mathcal{F}}
 \end{equation}
 \begin{equation}
   \mathcal{F}=\mathcal{F}_{N,M-1}\left( {\{ q^{\gamma_j} \} ,\{a_i \} \atop \{b_j q^{\gamma_j} \}};\{x_i \} \right).\quad 
    \end{equation}
  \begin{thm}\label{thm:laxj}
Specializing $a_M=1$ and setting $z=x_M$, $t=x_{M-1}$ and 
$\Psi_{j,i}=r_{j,i} \Psi_0\quad(1\leq j \leq N, 1\leq i \leq M-1)$,
the equations in Theorem \ref{thm:pfaff} can be reduced as 
\begin{equation}
  \begin{cases}\label{2kake2j}
  T_z \overrightarrow{\Psi}^{\rm red}=\overrightarrow{\Psi}^{\rm red} \displaystyle\biggl( \prod _{i=1}^{M-1} \frac{z-a_{M-i}x_{M-i}}{z-x_{M-i}}X_{2(M-i)-1} {X_{2(M-i)}}^{-1} \biggr){X_{2M-1}}^{-1}X_{2M}D_1,\\
  T_t \overrightarrow{\Psi}^{\rm red} = \overrightarrow{\Psi}^{\rm red}\displaystyle\frac{z-qt}{z-a_{M-1}qt}
  X_{2M-3}(z/q) {X_{2(M-1)}}^{-1}(z/q)D_2,
  \end{cases}
  \end{equation}
  where $\overrightarrow{\Psi}^{\rm red}$ is a vector with $N+1$ components:
   $\overrightarrow{\Psi}^{\rm red}=[\Psi_0^{\rm red},\Psi_{1,M}^{\rm red},\Psi_{2,M}^{\rm red},\cdots,\Psi_{N,M}^{\rm red}]$, $\Psi_{j,M}^{\rm red}=\Psi_{j,M}|_{a_M=1}$, 
  \begin{equation}
  \begin{array}{c}
D_1=\mathrm{diag}[1,q^{-\sum_{j=1}^N \gamma_j} ,q^{-\sum_{j=2}^N \gamma_j},\cdots,q^{-\gamma_N}], \quad
D_2=\mathrm{diag}[d_2,1,\cdots,1], \vspace{6pt}\\
X_k=\begin{bmatrix}
{u_{1,k}}&1&&\\
&{u_{2,k}}&\ddots&\\
&&\ddots&1\\
z&&&{u_{N+1,k}}
\end{bmatrix},\vspace{6pt}\\
\prod_{j=1}^{N+1}{u_{j,2i-1}}
=(-1)^Nx_{M-i},\quad \prod_{j=1}^{N+1}{u_{j,2i}}=
(-1)^Na_{M-i}x_{M-i}.
\end{array}
\end{equation}
Here $\{u_{j,i}\}$ are  independent of $z$ and they are given as a rational functions in $\{r_{j,i}\}$. Explicit forms of $\{u_{j,i} \}$ are given in the proof.
  \end{thm}
  \begin{pf}
  When $a_M=1$, obviously we have $T_z (\Psi_0)=\Psi_0$. 
  We will compute 
  $T_z(\Psi_{j,M})$ $(1\leq j \leq N)$.
In  the equation \eqref{appu}, 
we have 
\begin{equation}\label{zup}
\begin{array}{rll}
T_z \overrightarrow{\Psi}^{\rm red}=&\rho \mu \overrightarrow{\Psi}^{\rm red},\\
=&\rho (\overrightarrow{\Psi}^{(1)} Q),\\
=&\sigma_{M-1}\sigma_{M-2}\cdots \sigma_2  (\overrightarrow{\Psi}^{(2)}R_1Q^1),\\
=&\sigma_{M-1}\sigma_{M-2}\cdots \sigma_3  (\overrightarrow{\Psi}^{(3)}R_2R_1^2Q^1)\\
\vdots\\
=&\overrightarrow{\Psi}^{(M)}R_{M-1}R_{M-2}^{M-1}R_{M-3}^{M-1,M-2}\cdots R_1^{M-1,M-2,\cdots,2}\rho(Q),
\end{array}
\end{equation} 
where we put a vector $\overrightarrow{\Psi}^{(i)}:=\begin{bmatrix}\Psi_0&\Psi_{1,i}&\Psi_{2,i}&\cdots &\Psi _{N,i}\end{bmatrix}$ $(1\leq i \leq M)$ and $(*)^{k_1,k_2,\cdots,k_l}$ stands for $\sigma_{k_l}\sigma_{k_{l-1}}\cdots \sigma_{k_1}(*)$. 
The matrices $Q$ and $R_i$ come from the equations \eqref{ms} and \eqref{sgm}. 
The matrix $Q$ is a representation matrix of the transformation $\mu$ for the vector $\overrightarrow{\Psi}^{\rm red}$ and $\overrightarrow{\Psi}^{(1)}$. By the equation \eqref{ms}, we have
\begin{equation}\label{muq}
\begin{array}{rl}
\mu \overrightarrow{\Psi}^{\rm red}=\overrightarrow{\Psi}^{(1)}Q=&\overrightarrow{\Psi}^{(1)} 
{\scriptsize \begin{bmatrix}1&1&&&&\\&-b_1&\ddots&&&\\&&-b_2&\ddots&&\\&&&\ddots&\ddots&\\&&&&-b_{N-1}&1\\a_1x_1&&&&&b_N\end{bmatrix}^{-1}}
\begin{bmatrix}1&1&&&\\&-1&\ddots&&\\&&\ddots&&\\&&&-1&1\\x_1&&&&1\end{bmatrix}\\
&\quad  \times {\rm diag} [1,q^{-\sum_{j=1}^N \gamma_j},q^{-\sum_{j=2}^N \gamma_j},\cdots,q^{-\gamma_N}].
\end{array}
\end{equation}
The matrix $R_i$ is a representation matrix of the transformation $\sigma_i$ for the vector $\overrightarrow{\Psi}^{(i)}$ and $\overrightarrow{\Psi}^{(i+1)}$. 
In the first equation of the equation \eqref{sgm}, 
we put $r_{j,i}:=\Psi_{j,i}/\Psi_0$ and we have
\begin{equation}
\begin{array} {rl}
\sigma_i \overrightarrow{\Psi}^{(i)}&=\overrightarrow{\Psi}^{(i+1)} R_i\\
&=\overrightarrow{\Psi}^{(i+1)} \displaystyle\frac{a_ix_i-a_{i+1}x_{i+1}}{x_i-a_{i+1}x_{i+1}}\\
&\times {\scriptsize \begin{bmatrix}-x_{i}r_{N,i}&1\\
&-\frac{1}{r_{1,i}}&\ddots&\\
&&-\frac{r_{1,i}}{r_{2,i}}&\ddots&\\
&&&\ddots&&\\
&&&&1\\
a_{i+1}x_{i+1}&&&&-\frac{r_{N-1,i}}{r_{N,i}}\end{bmatrix}
\begin{bmatrix}
-a_ix_ir_{N,i}&1\\
&-\frac{1}{r_{1,i}}&\ddots\\
&&-\frac{r_{1,i}}{r_{2,i}}&\ddots\\
&&&\ddots&\hspace{-18pt}\ddots\\
&&&&1\\
a_{i+1}x_{i+1}&&&&-\frac{r_{N-1,i}}{r_{N,i}}
\end{bmatrix}^{-1}}
\end{array}
\end{equation}
  The second equation for $T_{t} \overrightarrow{\Psi}$ is obtained similarly by considering $T_{t}=\sigma_{M-1} T_z \sigma_{M-1}$. In fact, by the equations \eqref{0} and \eqref{gokan},  we have 
  \begin{equation}\label{te}
  \begin{array}{rl}
  T_t \overrightarrow{\Psi}&= \overrightarrow{\Psi} 
  \displaystyle\frac{z-qt}{z-qa_{M-1}t}
  \begin{bmatrix}
  v_1\frac{z-q a_{M-1}t}{z-qt}&\frac{v_2}{z-qt}&\cdots&\frac{v_{N+1}}{z-qt}\\
  &1&&\\
  &&\ddots&\\
  &&&1
  \end{bmatrix}\\
  &=\overrightarrow{\Psi} \displaystyle\frac{z-qt}{z-qa_{M-1}t}
\begin{bmatrix}
v_{1,1}&1&& \\
     &v_{2,1}&\ddots&\\
     &      &\ddots& 1\\
z/q&      &          &v_{N+1,1}
  \end{bmatrix}
  \begin{bmatrix}
v_{1,2}&1&& \\
     &v_{2,2}&\ddots&\\
     &      &\ddots& 1\\
z/q&      &          &v_{N+1,2}
  \end{bmatrix}^{-1}
  \begin{bmatrix}
  d_1'& &           &\\
       &1&          &\\
       &  &\ddots&\\
       &  &          &1
  \end{bmatrix}\\
  &=\overrightarrow{\Psi}\displaystyle\frac{z-qt}{z-qa_{M-1}t}{X'}_{2M-3}(z/q){X'}_{2(M-1)}^{-1}(z/q)D_2',
  \end{array}
  \end{equation}
  where 
  $\prod_{j=1}^{N+1}v_{j,1}=ta_{M-1}$, $\prod_{j=1}^{N+1}v_{j,2}=t$. And  $v_{j,i}$ are rational functions which are independent of $z$.
  \qed
  \end{pf}
\begin{rem}
The equations  \eqref{zup},  \eqref{te}  satisfy a compatibility condition by construction. Therefore, the rational function $r_{j,i}$ satisfies certain difference equation. For instance, in case of $(M,N)=(2,2)$, we have
\begin{equation}
\begin{array}{c}
\overline{r_{2,1}}=-\displaystyle\frac{q^{-\gamma _2} \left(a_1 t r_{2,1}+b_2
   r_{1,1}-b_1 r_{2,1}-r_{1,1}-b_2+1\right)}{-a_1 b_1 t r_{2,1}-a_1 t
   r_{1,1}+b_1 t r_{2,1}+t r_{1,1}+b_1 b_2-t},\\
   \overline{r_{1,1}}=
   \displaystyle\frac{q^{-\gamma _1-\gamma _2} \left(-a_1 b_1 t r_{2,1}-a_1 t
   r_{1,1}+a_1 t r_{2,1}+b_2 r_{1,1}+b_1 b_2-b_2\right)}{-a_1 b_1 t
   r_{2,1}-a_1 t r_{1,1}+b_1 t r_{2,1}+t r_{1,1}+b_1 b_2-t}.
   \end{array}
   \end{equation}
   These equations are interpreted as multivariable Riccatti equations for the special solution of the corresponding system $\mathcal{P}_{3,(2,2)}$. The solution is given in \eqref{rji}.
   \end{rem}
   In order to relate the above results to the monodromy preserving deformations $\mathcal{P}_{N+1,(M_,M+)}$, we specify the deformation equation as follows. First we choose $\vec{\varepsilon}$ as $\vec{\varepsilon}=(+,-,+,-,\cdots,+,-,-,+)$. Namely
   we consider a linear equation \eqref{lax}
\begin{equation}\label{laxa}
\begin{array}{c}
\Psi (qz)=\Psi(z)A(z), \ A(z)=DX_1(z)X_2^{-1}(z)\cdots X_{2M-3}(z)X_{2M-2}^{-1}(z)X_{2M-1}(z)^{-1}X_{2M}(z).\\
\end{array}
\end{equation}
   As a deformation of this, 
   we consider the  following equation
\begin{equation}\label{deb}
\overline{\Psi}=\Psi B(z),\quad B(z)=X_{2M}(z/q)X_{2M-1}^{-1}(z/q),
\end{equation}
where the  shift of parameters are given by 
\begin{equation}
\begin{array}{l}
\overline{a_k}=a_k\quad (k=1,2,\cdots,M),\\
\overline{b_j}=b_j\quad (j=1,2,\cdots,N),\\
\overline{c_l}= c_l\quad (l=1,2,\cdots,2M-2), \quad \overline{c_i}= q c_i\quad  (i=2M-1,2M),\\
\overline{d_j}=d_j \quad (j=1,2,\cdots,N).
\end{array}
\end{equation}
We denote by $\mathcal{P}_{N+1,(M,M)}$ the nonlinear equation arising as the comatibility condition $A(z)B(qz)=B(z)\overline{A(z)}$ of \eqref{laxa} \eqref{deb}. Then we have
    \begin{thm}
    Under the specialization $d_1=1$, 
  the nonlinear equation $\mathcal{P}_{N+1,(M,M)}$ \eqref{lax} admits a particular solution given in terms of a generalized hypergeometric functions $\mathcal{F}_{N,M-1}$ as follows.
    \begin{equation}\label{usol}
    \begin{array}{rl}
    &\begin{cases}
    u_{1,2(M-i)-1}=-\displaystyle\frac{r_{N,i}x_id_1}{d_{N+1}},\\
    u_{1,2(M-i)}=-\displaystyle\frac{r_{N,i}a_ix_id_1}{d_{N+1}},\\
    u_{1,2M-1}=\displaystyle\frac{d_1}{d_{N+1}},\\
    u_{1,2M}=\displaystyle\frac{d_1}{d_{N+1}},
    \end{cases} \left(1\leq i \leq M\right)\vspace{6pt}\\
    &\begin{cases}
    u_{j,2(M-i)-1}=-\displaystyle\frac{r_{j-2,i}d_j}{r_{j-1,i}d_{j-1}},\\
    u_{j,2(M-i)}=-\displaystyle\frac{r_{j-2,i}d_j}{r_{j-1,i}d_{j-1}},\\
    u_{j,2M-1}=-\displaystyle\frac{b_{j-1}d_{j}}{d_{j-1}},\\
    u_{j,2M}=-\displaystyle\frac{d_{j}}{d_{j-1}}, 
    \end{cases}\left({1\leq i \leq M\atop 2\leq j \leq N}\right)\vspace{6pt}\\
    &\begin{cases}
     u_{N+1,2(M-i)-1}=-\displaystyle\frac{r_{N+1-2,i}d_{N+1}}{r_{N,i}d_{N}},\\
    u_{N+1,2(M-i)}=-\displaystyle\frac{r_{N-1,i}d_{N+1}}{r_{N,i}d_{N}},\\
    u_{N+1,2M-1}=-\displaystyle\frac{b_{N}d_{N+1}}{d_{N}},\\
    u_{N+1,2M}=\displaystyle\frac{d_1}{d_N}
    \end{cases}(1\leq i \leq M),
    \end{array}
    \end{equation}
    and     \begin{equation}
    \begin{array}{c}
    c_{2i-1}=(-1)^Nx_{M-i}, c_{2i}=(-1)^N a_{M-i} x_{M-i}, \\
      d_1=1,d_k=q^{-\sum_{j=k-1}^N\gamma_j} \quad (1 \leq i \leq M, 2\leq k \leq N+1),
    \end{array}
    \end{equation} 
      where $r_{j,i}$ is a ratio of the hypergeometric function as given \eqref{rji}
  \begin{equation}
 r_{j,i}=\frac{1-b_j}{1-q^{\gamma_j}}\prod_{l=1}^j \frac{1-q^{\gamma_l}}{1-b_l q^{\gamma_l}}\prod_{k=1}^N \frac{(q^{\gamma_k},b_k)_{\infty}}{(b_k q^{\gamma_k},q)_{\infty}}\frac{T_{a_i}T_{x_1}T_{x_2}\cdots T_{x_{i-1}} T_{b_j}\mathcal{F}}{\mathcal{F}}
  \end{equation}
   and the function $\mathcal{F}$ stands for 
   $\mathcal{F}=\mathcal{F}_{N,M-1}\Bigl( {\{ q^{\gamma_j} \} ,\{a_i \} \atop \{b_j q^{\gamma_j} \}};\{x_i \} \Bigr)$
    $(1\leq j \leq N, 1 \leq i \leq M-1)$.
%
  \end{thm}
  \begin{proof}
Then, the system of equations  \eqref{2kake2j} is a  specialization of the equation \eqref{laxa} via a gauge transformation 
\begin{equation}
\widehat{\Psi}=\prod_{i=1}^{M-1}\frac{(a_{i}x_{i}q/z)_{\infty}}{(x_{i}q/z)_{\infty}}\Psi.
\end{equation}
In more detail, we define a transformation $\delta^{(\varepsilon_i,\varepsilon_{i+1})}_{k,l}$ of variables $\{ u_{j,k} \}$ and $\{ d_j\}$  
  such that 
  \begin{equation}
  \delta^{(\varepsilon_k,\varepsilon_l)}_{k,l}(DX_{k}^{\varepsilon_k}(z)X_{l}^{\varepsilon_l}(z))=X_k^{\varepsilon_k}(z) X_{l}^{\varepsilon_l}(z)D.
  \end{equation}
   This translation is uniquely given as follows
  \begin{equation}\label{dxx}
  \begin{array}{c}
  \delta^{\pm,\mp}_{k,l}(u_{j,k},u_{j,l},d_j)=(\displaystyle\frac{d_{j-1}}{d_j}u_{j,k},\displaystyle\frac{d_{j-1}}{d_j}u_{j,l},d_j).
  \end{array}
  \end{equation}
  We apply the transformation \eqref{dxx} on the equation \eqref{laxa} repeatedly. Then, the following equation holds
  \begin{equation}\label{xddx}
  \begin{array}{c}
  \delta_{2M-1,2M}^{-,+}\delta_{2M-3,2M-2}^{+,-}\cdots \delta_{2,3}^{+,-}\delta_{1,2}^{+,-}(DX_1X_2^{-1}\cdots X_{2M-3}X_{2M-2}^{-1}X_{2M-1}^{-1}X_{2M})\\
  =Y_1{Y_2^{-1}}\cdots Y_{2M-3}{Y_{2M-2}^{-1}}{Y_{2M-1}^{-1}}Y_{2M}D,
  \end{array}
  \end{equation}
  where a matrix $Y_i$ stands for  a matrix $X_i$ whose variable $u_{j,i}$ replaced $y_{j,i}=\displaystyle\frac{d_{j-1}}{d_j} u_{j,i}$. 
  Comparing the coefficients of the right-hand side of the equation \eqref{xddx} and of the first equation of the equation \eqref{2kake2j}, we obtain the following result \eqref{usol}
    where $d_1=1,d_k=q^{-\sum_{j=k-1}^N\gamma_j}\quad (2\leq k \leq N+1)$, $r_{0,i}=1$.
%
  We recall that the rational function $r_{j,i}$ is given as in \eqref{rji}
    \end{proof}
  \begin{rem}
   Theorem \ref{thm:laxj} gives a $q$-analogue of an argument to reduce a rank of a linear Pfaffian system in section 5 of \cite{tsuda15}.
  \end{rem}
  \begin{rem}\label{nmgar}
The system \eqref{lax} for $\mathcal{P}_{2,(M_-,M_+)}$ $(M_-+M_+=2M)$ is equivalent to the Lax equation for $2(M-1)$ dimensional $q$-Garnier system \cite{ny17} \cite{sakai05}  via the following gauge transformation
 \begin{equation}
 \widehat{\Psi}=\displaystyle\prod_{i=1}^{M_-}  
\left(\frac{(qc_i/z)_{\infty}}{(-z)_{\infty}(-q/z)_{\infty}}\right)
\Psi.
 \end{equation}
\end{rem}
\section*{Acknowledgement}
The author would like to express her gratitude to Professor Yasuhiko Yamada for valuable suggestions
and encouragement. 
She also thanks supports from JSPS KAKENHI Grant Numbers 17H06127 and 26287018 for the travel expenses in accomplishing this study. 
  
\end{document}